\documentclass{article}

\usepackage{arxiv}
\usepackage[utf8]{inputenc} 
\usepackage[T1]{fontenc}    
\usepackage{hyperref}       
\usepackage{url}            
\usepackage{booktabs}       
\usepackage{amsfonts}       
\usepackage{nicefrac}       
\usepackage{microtype}      
\usepackage{graphicx}
\usepackage{natbib}
\usepackage{doi}

\usepackage{times}
\usepackage{soul}
\usepackage{url}
\usepackage{enumitem}
\usepackage{graphicx}
\usepackage{amsmath}
\usepackage{amsthm,bbm}
\usepackage{booktabs}
\usepackage{thm-restate}
\usepackage{multicol}
\usepackage[linesnumbered,ruled,vlined]{algorithm2e}
\usepackage{algorithmic}
\usepackage{xcolor}
\usepackage[capitalise]{cleveref}
\usepackage{empheq}
\usepackage{verbatim}

\newcommand*\widefbox[1]{\fbox{\hspace{2em}#1\hspace{2em}}}
\graphicspath{ {./images/} }
\SetKwInput{KwInput}{Input}                
\SetKwInput{KwOutput}{Output}              
\newtheorem{example}{Example}
\newtheorem{theorem}{Theorem}
\newtheorem{lemma}[theorem]{Lemma}

\theoremstyle{definition}
\newtheorem{definition}{Definition}

\newcommand{\ourname}{$\alpha$-\textsf{fairness }}
\newcommand{\cof}{\textsc{CoF}}
\newcommand{\OurFPP}{\textsc{LinP-FFP}}

\SetCommentSty{mycommfont}
\SetKwComment{Comment}{$\triangleright$\ }{}

\title{ Individual Fairness in Feature-Based Pricing for Monopoly Markets }

\author{ Shantanu Das \\
	\href{https://mll.iiit.ac.in/}{Machine Learning Lab}\\
	IIIT Hyderabad, India\\
	\texttt{shantanu.das@research.iiit.ac.in} \\
	\And
	Swapnil Dhamal \\
	T\'el\'ecom SudParis \\
	\'Evry, France \\
	\texttt{swapnil.dhamal@gmail.com} \\
	\AND
	Ganesh Ghalme \\
	Technion \\
	Israel Institute of Technology, Israel\\
	\texttt{ganeshghalme@gmail.com} \\
	\And
	Shweta Jain \\
	Indian Institute of Technology \\
	Ropar, India \\
	\texttt{shwetajain@iitrpr.ac.in} \\
	\And
	Sujit Gujar \\
	\href{https://mll.iiit.ac.in/}{Machine Learning Lab}\\
	IIIT Hyderabad, India\\
	\texttt{sujit.gujar@iiit.ac.in} \\
}

\date{}

\renewcommand{\headeright}{}
\renewcommand{\undertitle}{}
\renewcommand{\shorttitle}{ Individual Fairness in Feature-Based Pricing for Monopoly Markets }

\hypersetup{
pdftitle={ Individual Fairness in Feature-Based Pricing for Monopoly Markets },
pdfsubject={},
pdfauthor={Shantanu Das, Swapnil Dhamal, Ganesh Ghalme, Shweta Jain, Sujit Gujar},
pdfkeywords={},
}

\begin{document}
\maketitle

\begin{abstract}
We study fairness in the context of feature-based price discrimination in monopoly markets. We propose a new notion of individual fairness, namely, \ourname, which guarantees that individuals with similar features face similar prices. First, we study discrete valuation space and give an analytical solution for optimal fair feature-based pricing.  We show that the cost of fair pricing is defined as the ratio of expected revenue in an optimal feature-based pricing to the expected revenue in an optimal fair feature-based pricing  (\cof) can be arbitrarily large in general.
When the revenue function is continuous and concave with respect to the prices, we show that one can achieve \cof\  strictly less than $2$, irrespective of the model parameters. Finally, we provide an algorithm to compute fair feature-based pricing strategy that achieves this \cof. 
\end{abstract}

\section{Introduction}
\label{sec:intro2}
The Internet has transformed the way markets function. 
Today's Internet-based ecosystems such as entertainment and e-commerce marketplaces are more consumer-centric and information-driven than ever before. Data and  AI  systems are primarily used to power advertising, consumer retention, and personalized experience. These AI systems are deployed to aggregate individual choices and preferences to make personalized experiences possible. It is a common practice to use aggregated information about consumers to offer different prices to different consumers or segments of the market; this practice is commonly termed \emph{price discrimination}~\cite{varian92}. 

Price discrimination has come under ethical scrutiny on multiple instances in the recent past. For example, it was found that  Orbitz, an online travel agency, charges Mac users more than Windows users~\cite{Orbitz}. Uber's strategy to charge personalized prices came under heavy consumer backlash~\cite{uberOne,uberTwo}, and thanks to the fine-grained data analysis of consumer behavior, several such instances were reported in e-commerce and retail industry~\cite{Hinz11}. More recently, \cite{akshat2021} showed that neighborhoods
with high non-white populations, higher poverty, younger
residents, and high education levels faced higher cab trip
fares in Chicago.
Not surprisingly, the regulatory bodies and research community has taken notice. Economists have raised concerns on fairness issues of personalized pricing~\cite{stefan}.  
Price discrimination based on nationality or residence is made illegal in the ~\cite{EU-PD}. In the USA, a white house report provides guidelines for enforcing existing anti-discrimination, privacy, and consumer protection laws while practicing discriminatory pricing~\cite{usa}. Given the overwhelming evidence and rising concerns, there is an urgent need to study price discrimination and fairness formally.

Sellers or firms use price discrimination for multiple reasons, including increasing revenue, covering transportation and storage costs, increasing market reach, rewarding loyal consumers, promoting a social cause, and so on \cite{cassady46}. In general,  price discrimination does not always raise ethical and fairness issues and hence requires a careful inspection to categorize situations where this practice may lead to treatment disparity and invite regulatory intervention \cite{kaplan}.
In this work, we focus on designing the pricing strategies for a seller (monopolist) who wants to maximize the revenue via price discrimination while ensuring fairness amongst the consumers.
 
A revenue-maximizing seller with complete knowledge of consumer valuations without fairness consideration would charge each consumer her valuation for the product. This pricing strategy, otherwise called \emph{first-degree} price discrimination, may result in wild fluctuations in prices and is considered unfair in general \cite{Moriarty21}. Also, in practice,  sellers do not have full access to individual consumer valuations but may have a  distribution over valuations through \emph{features}. In such 
\emph{feature-based pricing} (FP), 
the seller segregates the market into segments through the consumer features. 
The seller's problem then reduces to finding optimal pricing for each segment~\cite{bergemann2015limits,algorithmic-pricedisc}. Such FP is referred to as \emph{third-degree price discrimination} in the literature. 
In this paper, our goal is to ensure fairness issues in feature-based personalized pricing.

\paragraph{Our Contributions} 
We introduce the notion of \emph{\ourname} in price discrimination which ensures that similar individuals face similar prices. We emphasize that if individuals with similar features are charged differently by segregating them into different segments, the interpersonal price comparison based on their features renders fairness issues. With this, we introduce a model for optimal \emph{fair feature-based pricing} (FFP) as the problem of maximizing revenue while ensuring \ourname$\!\!$. We begin with two segments in the market and discrete valuations and propose an optimal FFP scheme (Section \ref{ssec:optdiscreteffp}). To quantify the loss in the revenue due to fairness, we then introduce \emph{cost of fairness} (\cof) -- the ratio of expected revenue in an optimal  FFP to the expected revenue in an optimal FP. We prove that a constant lower bound on \cof\ is impossible to achieve in general. 

Next, in Section \ref{sec:ffpcontinuous}, under the assumption that  the  revenue function is concave in offered prices
\cite{bergemann2021thirddegree}\footnote{this assumption is standard in economics as a large number of probability distributions follow this},
we show that one can achieve a constant upper bound on \cof. Here, first, we show that the seller can compute optimal  FFP using a convex program if it has access to distributional information (knows all consumers' valuation distribution functions). We then identify a class of FFP strategies, namely \OurFPP\ that satisfy \ourname$\!\!$. With the help of these pricing strategies, we then show that the  \cof\ is strictly less than $2$ irrespective of model parameters. 
Finally, we propose OPT-\OurFPP, an  $O(K\log(K))$ time algorithm where $K$ is the number of segments that does not need access to complete distributional information and computes $\alpha$-fair pricing that achieves the aforementioned \cof\  (Algorithm~\ref{alg:ouralgo} and Theorem~\ref{thm:opt_cofbound}).

\section{Related Work}

The impact of discriminatory pricing on consumer and seller surplus was first considered by \cite{bergemann2015limits} when the consumer characteristics are known to the seller. The authors proposed a method to provide the optimal market segmentation. The generalized problem was then considered by \cite{algorithmic-pricedisc} which extended the work of \cite{bergemann2015limits} to the case where only partial information about the consumer's valuation was known to the seller.

When the valuations of the consumers are not known, \cite{value-personalized,value-personalpricing} propose feature-based pricing and provides bounds on the value generated using idealized personalized pricing and Feature-based pricing over Uniform pricing. The value of feature-based pricing depends on the correlation of valuations and consumer features. \cite{random-network} consider the first-degree price discrimination over the social network where the centrality measures in social networks determine the features of the consumers. They provide bounds on the value of network-based personalized pricing in large random social networks with varying edge densities.
Our work follows a similar approach because we derive personalized pricing from the features. However, naive feature-based pricing can be very unfair to the consumers, as we show in Proposition \ref{prop:discretecof}. Our focus is to design feature-based pricing that is fair at the same time. 

Recently, many questions have been raised on the ethical side of price discrimination methods. \cite{Moriarty21} strongly criticizes online personalized pricing and suggests that personalized prices compete unfairly for social surplus created by transactions. \cite{ethical-legal} points out the need to design personalized pricing with ethical considerations, which can provide win-win outcomes for both organizations and consumers. \cite{fairness-tackled} discusses that discriminatory pricing leads to the perception of unfairness amongst the consumers, which undermines the stability of retail platforms. They discuss that when consumers are involved in forming the prices, this leads to improved fairness perception, thus leading to better retentivity. \cite{design-against-discrimination} discusses that web-based platforms typically use many private features of user profiles to connect buyers and sellers. When users interact on such platforms, it leads to discrimination regarding race, gender, and possibly other protected characteristics. All these studies lead to understanding the optimal price discriminatory strategies under the fairness constraint, which is the focus of our work. 

Finally, \cite{personalfairness} presents a list of metrics like \emph{price disparity}, \emph{equal access}, \emph{allocative efficiency fairness} to measure and analyze fairness in feature-based pricing and study its interplay with welfare. The metrics discussed are mainly the group fairness notions which are entirely different from \ourname discussed in this paper. We emphasize that though the above papers discuss the ethical issues in price discrimination, none of them provides a systematic approach to design the pricing strategy that maximizes the revenue and ensures the fairness guarantee.

\section{Preliminaries}
\label{sec:prelim}
We consider a  market with a monopolist seller  seeking to price a single product available in infinite supply.  The market is divided into finite number of segments $\mathcal{X} = \{x_1, x_2, \dots , x_K\}$, where $x_i$ represents the $i^{\text{th}}$ segment.
The seller, given access to $\mathcal{X}$, can choose to price  discriminate across segments to extract maximum revenue.

\begin{table*}[ht!]
    \centering
    \renewcommand{\arraystretch}{1.2}
    \begin{tabular}{|c|c|}
        \hline \textbf{Notation} & \textbf{Description} \\
        \hline 
        FP & Feature-based Pricing \\ \hline
        FFP & Fair Feature-based Pricing \\
        \hline
        $\mathcal{F}_k$, $f_k()$ &  Valuations CDF, PDF for $k^{\text{th}}$ consumer segment respectively\\ \hline
        $\mathcal{X}$ & Set of all consumer features/types \\ \hline
        $\mathcal{V}$ & Support set of consumers' valuations \\ \hline
        $x_k$ & Consumer feature of the $k^{\text{th}}$ segment \\ \hline
        $\beta_k$ & The fraction of consumers in the $k^{\text{th}}$ segment \\ \hline
        $\mathbf{p} = (p_1, p_2, \ldots p_K)$ & Feature-based price vector \\ \hline
        $\pi_k(p_k)$ & Revenue generated per consumer in the $k^{\text{th}}$ segment \\ \hline 
        $\Pi(p)$ & Revenue generated by $p$ across all consumer segments \\ \hline
        $\mathbf{\widehat{p}} = (\widehat{p}_1, \widehat{p}_2, \ldots \widehat{p}_K)$ & Price function in optimal price discrimination \\ \hline
        $d_{ij}:=d(x_i, x_j)$ & A real-valued metric on the consumer feature space $\mathcal{X}$ \\ \hline
        $\alpha$ & Fairness parameter \\ \hline
        $\mathbf{p^{\star}} = (p^{\star}_1, p^{\star}_2, \ldots p^{\star}_K)$ & Optimal fair feature-based price function \\ \hline
        $\widetilde{\mathbf{p}} = (\widetilde{p}_1,\widetilde{p}_2,\ldots,\widetilde{p}_K)$ & Price vector for OPT-\OurFPP \\ \hline
        \cof & Cost of Fairness \\ \hline
        $L_m$ & Linear approximation of concave revenue curve with $m$ as parameter \\ \hline
    \end{tabular}
    \renewcommand{\arraystretch}{1}
    \caption{Notation Table}
    \label{tab:notations}
\end{table*}

Consumers' valuations for the single product are non-negative random variables drawn from the set $\mathcal{V}$ (same across all segments). Let  $\mathcal{F}_i(\cdot)$ be the cumulative distribution function for the valuation of the consumers in $i^{\text{th}}$ segment, and  $f_i(\cdot)$ be corresponding probability density function (probability mass function when $\mathcal{V}$ is discrete).
In this paper, we consider the following two cases separately,  (a) $\mathcal{V}$ is discrete and finite, and (b) $\mathcal{V}$ is continuous.
Next, we present feature-based pricing model.

\subsection{Feature-based Pricing Model}
\label{ssec:featureprice} In feature-based pricing (FP), one can consider, without loss of generality, that the 
consumer feature is a representative of the market segment to which she belongs. Note that multiple consumers may have the same feature vector, and all the consumers having identical features belong to the same market segment. 
For simplicity, we will write $p_i:=$ price offered to the consumer in the $i^{\text{th}}$ segment. 
A consumer makes the purchase only if her valuation is equal to or more than the offered price. 
The expected revenue  per consumer generated from the $i^{\text{th}}$ segment with a price $p_i\in \mathbb{R}_{+}$  is given by 
\begin{equation}
    \pi_i(p_i) = p_i \cdot (1 - \mathcal{F}_i(p_i)) \\
\end{equation}
Whenever it is clear from the context we refer to expected revenue  per consumer from a segment to be expected revenue from that segment. Let $\beta_i$ be the fraction of consumers in the $i^{\text{th}}$  segment, then the expected revenue per consumer generated across all  segments is given as
$\Pi(\mathbf{p}) = \sum_{x_i \in \mathcal{X}} \beta_i\pi_i(p_i)
$.
We assume that $\beta_i$s are known to the seller. 
We call the sellers problem of revenue maximization as OPT$_{FP}(\mathcal{V,\mathcal{X},F,\beta})$ where $\mathcal{F}=(\mathcal{F}_1,\ldots,\mathcal{F}_K)$ and $\beta=(\beta_1,\ldots,\beta_K)$.

In the absence of fairness constraints, OPT$_{FP}(\cdot)$  reduces to charging each segment separately and optimal FP strategy $\widehat{\mathbf{p}}$ consisting $\widehat{p_i}$ for segment $i$ is  given by
 $
    \widehat{p}_i \in   \underset{p_i \in \mathbb{R}_{+} }{\operatorname{argmax}} \  \pi_i(p_i).$
 
\paragraph{Fairness in Feature-based Pricing}
\label{ssec:ind-fair}

Let $d:\mathcal{X} \times \mathcal{X} \rightarrow \mathbb{R}_{+}$ be  a distance function over $\mathcal{X}$. We assume that such a function exists and is well defined in $\mathcal{X}$, i.e., $(\mathcal{X},d)$ is a metric space. The distance function   quantifies the   dissimilarity between feature vectors of individuals belonging to  market segments.
For simplicity we write $d(x_i, x_j) := d_{ij}$. Individual fairness in FP strategy is defined as:
\begin{definition}[\ourname$\!\!$]
\label{def:fairness} 
A price function $\mathbf{p}:\mathcal{X} \rightarrow \mathbb{R}_+^{K}$ is $\alpha$-fair with respect to  $d$ iff for all $x_i,x_j \in \mathcal{X}$, we have
\begin{equation}\label{eq:alphafair}
 | p_i - p_j| \leq \alpha \cdot  d_{ij}.
\end{equation}
\end{definition}
We call a  pricing  strategy Fair Feature-based Pricing ($\alpha$-FFP) that  satisfies \cref{eq:alphafair} with a given value of $\alpha$. It is easy to see from the definition that any $\alpha$-FFP is also $\alpha{'}$-FFP for any $\alpha{'} \geq  \alpha$. We will drop the quantifier $\alpha$ and call it FFP when it is clear from the context. 

\paragraph{Cost of Fairness (\cof) } Next, we define \cof\  as the deviation from optimality due to fairness constraints given in \cref{eq:alphafair}. It is defined as the ratio of expected revenue generated by optimal feature-based pricing and fair feature-based pricing. 

\begin{definition}[\textsc{Cost of Fairness }(\cof)]
Cost of fairness for an FFP strategy $\mathbf{p}$ is defined as
\begin{equation}
    \cof = \frac{\Pi(\mathbf{\widehat{p}})}{\Pi(\mathbf{\mathbf{p}})}. 
\end{equation}
\end{definition}

In the following sections, we analyze FP and FFP strategies and their \cof\ when $\mathcal{V}$ is discrete (\cref{sec:discrete})
and continuous (\cref{sec:cont}).

\section{FFP for Discrete Valuations}\label{sec:discrete}
We want to ensure \ourname\ in the pricing strategy given the optimal FP. \ourname\ is achieved by maximizing revenue while satisfying the fairness constraints. In this section, we derive optimal FP (\cref{ssec:optdiscretefp}), propose how to achieve \ourname\ (\cref{ssec:optdiscreteffp}), and provide an upper bound on \cof\ (\cref{ssec:cofdiscrete}) for discrete valuation setting.

We consider the simplest setting described as follows:
Let the consumer segments be given by $\mathcal{X} = \{ x_1, x_2\}$ and their valuations are drawn from a discrete set $\mathcal{V} = \{ v_1, v_2\}$, we assume $v_1 < v_2$ without loss of generality. 
Let $\beta_1 = \beta$ and $\beta_2 = 1-\beta$. Further, let $f_1(v_1) = q_1$ ($f_2(v_1) = q_2$) denote the probability that a consumer has valuation $v_1$ in segment 1 (segment 2).
The expected revenue generated by $\mathbf{p}$ is given by:
\begin{align}
\label{eq:opt-price}
\Pi(\mathbf{p})=&\beta p_1 [ q_1 \mathbbm{1}(v_1 \geq p_1) + (1-q_1) \mathbbm{1}(v_2 \geq p_1) ]
\nonumber \\ 
& + (1-\beta) p_2 [ q_2 \mathbbm{1}(v_1 \geq p_2) + (1-q_2) \mathbbm{1}(v_2 \geq p_2) ]
\end{align}

\subsection{Optimal Feature-based Pricing} \label{ssec:optdiscretefp}
As discussed earlier, $\Pi(\mathbf{p})$ can be maximized by maximizing $\pi_i(p_i)$ for each market segment independently if there are no fairness constraints.
This problem is an integer program with price for each consumer type being a discrete variable. 
The revenue generated depends on $\beta_i$ and $f_i(\cdot)$ ($\beta$, $q_1$, $q_2$ in the current simplest case). The optimal FP is then given as 
\begin{equation} \label{eq:idealprice1}
    \text{For }\, i \in \{1,2\}:\; \widehat{p}_i = \begin{cases} 
      v_1 & \text{if } q_i\geq 1-\frac{v_1}{v_2} \\
      v_2 & \text{otherwise} 
   \end{cases}
\end{equation}
\begin{proof}
For a market segment $i$, $\pi_i(v_1) = v_1$ and $\pi_i(v_2) = v_2(1-q_i)$. So, $\widehat{p}_i = v_1$ if 
\begin{equation*}
\pi_i(v_1) \geq \pi_i(v_2) \implies v_1 \geq v_2(1-q_i) \implies q_i \geq 1 - \frac{v_1}{v_2}
\end{equation*}
otherwise, $\widehat{p}_i = v_2$.

\end{proof}
Next, we analyze the fairness aspects of the above pricing strategy.

\subsection{Optimal Fair Feature-based Pricing} \label{ssec:optdiscreteffp}
Let $(\mathcal{X},d)$ be a metric space. We model the Optimal fair feature-based pricing (FFP) problem as integer program which maximizes $\Pi(\mathbf{p})$ with \ourname constraints described in Eq.\eqref{eq:alphafair}. We denote this problem as OPT$_{FFP}(\mathcal{V}, \mathcal{X}, d, \mathcal{F}, \mathbf{\beta}, \alpha)$ and the corresponding optimal FFP strategy is denoted as $\mathbf{p}^{\star}$. First we make an interesting and very useful claim for binary valuations. 

\begin{lemma}
When $\mathcal{V}=\{v_1,v_2\}$, and if  $\widehat{\mathbf{p}}$ is not $\alpha$-fair, 
OPT $_{\mbox{FFP}} (\mathcal{V},\mathcal{X},d,\mathcal{F}, \mathbf{\beta},\alpha)$ reduces to OPT$_{\mbox{FP}} (\widetilde{\mathcal{V}},\mathcal{X},\mathcal{F}, \mathbf{\beta}) $ where $\widetilde{\mathcal{V}} \text{ is either } \{v_1\}, \text{ or } \{v_2\}, \text{ or } \{v_1, v_1+\alpha d_{12}\} $.

\label{lem:reduction}
\end{lemma}
\begin{proof}
Let $(p_1,p_2)$ be the tuple of offered prices. Note that if $v_2 - v_1 \leq  \alpha d_{12}$ or $\widehat{p}_1 = \widehat{p}_2$,  then the optimal $\mathbf{p}^{\star} = \mathbf{\widehat{p}}$ with support $\{v_1, v_2\}$ and $\mathbf{\widehat{p}}$ will be trivially fair. We consider a more interesting case when $v_2 - v_1 >  \alpha d_{12}$ and $\widehat{p}_1 \ne \widehat{p}_2$. In this case, the only candidate support sets for optimal fair pricing strategy are: $\{ v_1\}$, $\{ v_2\}$, $\{ v_1, v_1 + \alpha d_{12}\}$, $\{v_2-\alpha d_{12}, v_2\}$. The optimal FFP does not take values from the set $\{v_2-\alpha d_{12}, v_2\}$ as the consumers with valuation $v_1$ would not make any purchase. Hence, the expected revenue with support $\{v_2-\alpha d_{12}, v_2\}$  will be less than or equal to the expected revenue with support $\{v_2\}$.
\end{proof}

We now relax the constraint of binary valuation and analyze the optimal fair pricing scheme for $n$ valuations. The consumer segments are $\mathcal{X} = \{x_1,x_2\}$ with $\beta_1 = \beta$ and $\beta_2 = 1-\beta$, the valuations are drawn from the set $\mathcal{V} = \{v_1, v_2, \dots , v_n\}$, and $f_1(v_i) = q_{i,1}$ and $f_2(v_i) = q_{i,2}$. 
This is a simple extension of the pricing problem, OPT$_{FP} (\mathcal{V},\mathcal{X},\mathcal{F}, \mathbf{\beta})$ modelled as an integer program where the prices are drawn from the set $\mathcal{V}$. If $\mathbf{\widehat{p}}$ is not $\alpha$-fair then, the corresponding OPT$_{FFP} (\mathcal{V},\mathcal{X},d,\mathcal{F}, \mathbf{\beta}, \alpha)$ can be solved by reducing it to OPT$_{FP} (\mathcal{\widetilde{V}},\mathcal{X},\mathcal{F}, \mathbf{\beta})$ with $\widetilde{\mathcal{V}}$ given by:

\begin{equation*}
    \widetilde{\mathcal{V}} = \begin{cases} 
      \{v_i\}, v_i \in \mathcal{V}  & \text{if } p_1^{\star} =  p_2^{\star} \\
      \{v_j, v_j+\alpha d_{12}, v_j-\alpha d_{12}\}, v_j \in \mathcal{V} & \text{if } p_1^{\star} \neq p_2^{\star} 
   \end{cases}
\end{equation*} 
Given the set $\mathcal{\widehat{V}}$, the pricing problem OPT$_{FP} (\mathcal{\widetilde{V}},\mathcal{X},\mathcal{F}, \mathbf{\beta})$ can be solved in constant time.  It is easy to see that computing $\mathcal{\widehat{V}}$ takes $\mathcal{O}(n^{2})$ time for $n$ valuations and $2$ consumer types. Therefore, the fair pricing problem OPT$_{FFP} (\mathcal{V},\mathcal{X},d,\mathcal{F}, \mathbf{\beta}, \alpha)$ can be solved in $\mathcal{O}(n^{2})$ time.

\subsection{\cof\  Analysis} \label{ssec:cofdiscrete}
For $n=2$, based on the values of $q_1, q_2$ we have the following cases:
\begin{multicols}{2}
\begin{enumerate}[leftmargin=*]
    \item $p_1^{\star} = p_2^{\star} = v_1$
    \item $p_1^{\star} = p_2^{\star} = v_2$
    \item $p_1^{\star} = v_1 + \alpha d_{12}$, $p_2^{\star} = v_1$
    \item $p_1^{\star} = v_1$, $p_2^{\star} = v_1 + \alpha d_{12}$ 
\end{enumerate}
\end{multicols}
\noindent
In cases 1 and 2, optimal fair pricing is equivalent to uniform pricing and therefore are `trivially' fair with \cof\  = 1, i.e., $\Pi(\widehat{\mathbf{p}}) = \Pi(\mathbf{p}^{\star})$.
For case 3, $\Pi(\widehat{\mathbf{p}})$ and $\Pi(\mathbf{p}^{\star})$ are given as:  
\begin{align*}
    & \Pi(\widehat{\mathbf{p}}) = \beta(v_2)(1-q_1) + (1-\beta)v_1 \\
    & \Pi(\mathbf{p}^{\star})  = \beta(v_1 + \alpha d_{12})(1-q_1) + (1-\beta)v_1
\end{align*}
Then the cost of fairness for case 3 is given as:

\begin{align}
    \cof\  & = \frac{\Pi(\widehat{\mathbf{p}})}{\Pi(\mathbf{p}^{\star})} = \frac{\beta(v_2)(1-q_1) + (1-\beta)v_1}{\beta(v_1 + \alpha d_{12})(1-q_1) + (1-\beta)v_1} \nonumber \\
    & = \frac{\beta(v_2-v_1) + v_1 - \beta v_2 q_1}{\beta \alpha d_{12}(1-q_1) - \beta v_1 q_1 + v_1 } \nonumber \\
    & = \frac{\beta \left( 1 - \frac{v_1}{v_2}\right) + \frac{v_1}{v_2} - \beta q_1}{\beta \left (\frac{\alpha d_{12}}{v_2}\right)(1-q_1) - \beta \left( \frac{v_1}{v_2}\right)q_1 + \frac{v_1}{v_2}}
    \label{eq:cof}
\end{align}

Replacing $\beta$ with $(1-\beta)$ and $q_1$ with $q_2$ in the above expression, we get a similar approximation of \cof\  for case 4.

\begin{restatable}{proposition}{} Cost of fairness with discrete valuations can go arbitrarily bad. 
\label{prop:discretecof}
\end{restatable}
\begin{proof}

From \cref{eq:cof}  when $\frac{v_1}{v_2} \rightarrow 0$, we have 
$    \cof  = \frac{v_2}{\alpha d_{12}}$. 
The \cof\ (in Case 3 and/or Case 4) is arbitrarily bad if $d_{12}>0$ when there is a large difference between $v_1$ and $v_2$. Note that $d_{12}=0$ is uninteresting as the seller is unable to distinguish between two segments. \end{proof}

Note that $v_2$ being arbitrarily large need not be a commonly occurring setting. Hence, we work with bounded support valuations in the backdrop of the above negative results. 
In the next section, we make assumptions based on standard economic literature about the revenue functions $\pi_i(\cdot)$, i.e., concave revenue functions and common support~\cite{bergemann2021thirddegree}.  As argued in Section 3 of \cite{dhangwatnotai2015}, valuation distributions satisfying Monotone Hazard Rate (MHR) satisfy the assumptions as mentioned above regarding revenue functions.
It is also observed that the revenue functions are concave for another commonly analyzed family of distributions in literature called the regular distributions in which the virtual valuation is non-decreasing (Section 4.3 of \cite{bergemann2021thirddegree}).
MHR is a common assumption in Econ-CS~\cite{hartline2009}.
Therefore, in the following section, we analyze the cost of fairness for such valuation distributions and the associated concave revenue functions.

\section{FFP for Continuous Valuations}
\label{sec:cont}
In this section, we consider feature-based pricing with continuous valuations. We impose a standard restriction on the revenue functions $\pi_i(\cdot)$ such that they are concave on the common support $\mathcal{V} = [\underline{v}, \Bar{v}]$ \cite{bergemann2021thirddegree}.  
The consumer segments are identified by the associated feature vectors $x_i \in \mathcal{X}$. 
$\underline{v}$ is the marginal cost defined as a minimum feasible valuation for which a seller is willing to sell the product. The marginal cost may include the cost of production, transportation, etc. On the other hand, $\Bar{v}$ is the maximum consumer valuation. Without loss of generality, we consider that maximum consumer valuation is greater than marginal cost; i.e., trade occurs.       

We begin with a tight upper bound on the \cof\ under conditions as mentioned above (\cref{sec:ffpcontinuous}) followed by two pricing schemes based on the available information about the revenue functions (\cref{sec:ourffp}), and finally, we present an algorithm that achieves the \cof\ bound in \cref{sec:algorithm}.

\subsection{Optimal FFP for Continuous Valuations}
\label{sec:ffpcontinuous}
The problem of determining optimal FFP can be modeled as a convex program with \ourname\ as linear constraints.
The convex program below describes OPT$_{FFP}(\mathcal{V}, \mathcal{X},d, \mathcal{F},\mathbf{\beta},\alpha)$ model with complete knowledge of revenue functions $\pi_i(\cdot)$.
\begin{subequations}
\begin{empheq}[box=\widefbox]{align*}
    \max_{p_k \in \mathcal{V}, \forall k}  \Pi(\mathbf{p}) &= \sum_{k=1}^{K} \beta_k \pi_k(p_k)\\
\text{subject to, } 
    |p_i - p_j| &\leq \alpha d(x_i, x_j), \forall i \neq j \\
    & p_i \geq 0, \forall i \in [K]
\end{empheq}
\end{subequations}
Let $\mathbf{p^{\star}}$  be a solution to the above problem.

\subsection{\OurFPP\ and \cof\  Analysis}
\label{sec:ourffp}
Let $D_i := \min_{j\neq i} d_{ij}$.  With the following proposition, we propose a class of $\alpha$-fair pricing strategies.
\begin{restatable}{proposition}{ClaimOne} For a given  $m\in[\underline{v},\overline{v}]$, if the price function satisfies
$|p_i - m| \leq    \frac{\alpha}{2} D_i$ for all $i \in [K]$ then  it satisfies \emph{\ourname}.
\label{prop:linp}
\end{restatable}
\begin{proof}
From triangle inequality, we have 
$ |p_i - p_j|  \leq |p_i - m| + |p_j - m|
     \leq \frac{\alpha}{2} D_i + \frac{\alpha}{2} D_j \leq \alpha d_{ij}$.
The last inequality results from the fact that $ D_i = \min_{k\neq i}d_{ik} \leq d_{ij} $ and $ D_j = \min_{k\neq j}d_{ik} \leq d_{ji} = d_{ij} $.
\end{proof}

In other words, for ensuring that the prices for different segments are not too different, it is enough to ensure that the pricing for each segment is not too different from some common point $m$. The pricing for all the segments would hence be around this point and could be determined with respect to this point. We term this point as {\em pivot}.
We now present the second FFP model, an $\alpha$-fair pricing strategy that is pivot-based and satisfies the condition in \cref{prop:linp}, with access to only $\widehat{p}_i$ for a given $m$.  
\begin{equation}
    p_i = \begin{cases} m + \alpha D_i/2   & \text{ if } \  \widehat{p}_{i} - m \geq \alpha D_i/2 \\
    m - \alpha D_i/2 & \text{ if } m - \widehat{p}_{i}  \geq  \alpha D_i/2 \\ 
    \widehat{p}_i & \text{ otherwise}   \end{cases}
    \label{eq:fairPricing}
\end{equation}
We call this pricing scheme \OurFPP.
It is easy to see that the above pricing strategy is $\alpha$-fair. We now present the \cof\ bound for \OurFPP.

\begin{theorem} \label{thm:cofbound}
The Cost of Fairness for optimal fair price discrimination with concave revenue functions satisfies
$$\cof\ \leq \frac{2}{1+
\min \left\{ \alpha \frac{\min_i D_i}{\Bar{v} - \underline{v}} , 1 \right\}}
$$
\end{theorem}
\begin{proof}
We prove that the above \cof\ is satisfied by \OurFPP\ and hence the theorem. 
Let $m \in [\underline{v}, \Bar{v}]$ be a pivot point (See Figure \ref{fig:my_label}). Let 
\begin{equation}\label{eq:gamma}
    \gamma_i:= \begin{cases}
    \frac{(m-\underline{v}) + \alpha D_i/2}{\widehat{p}_i - \underline{v}} & \text{if } \widehat{p}_{i} - m \geq \alpha D_i/2  \\ 
    \frac{ (\Bar{v}-m) + \alpha D_i/2  }{\Bar{v} - \widehat{p}_i} & \text{if } m - \widehat{p}_{i} \geq \alpha D_i/2 \\
    1 & \text{ otherwise}
\end{cases}\end{equation}
Let $\widehat{\pi}_{i}$ be the expected revenue generated from the $i^{\text{th}}$ segment under $\widehat{\mathbf{p}}$. We now show the following supporting lemma.
 
\begin{figure}
\label{im:linapprox}
    \centering
\includegraphics[width=.85\columnwidth]{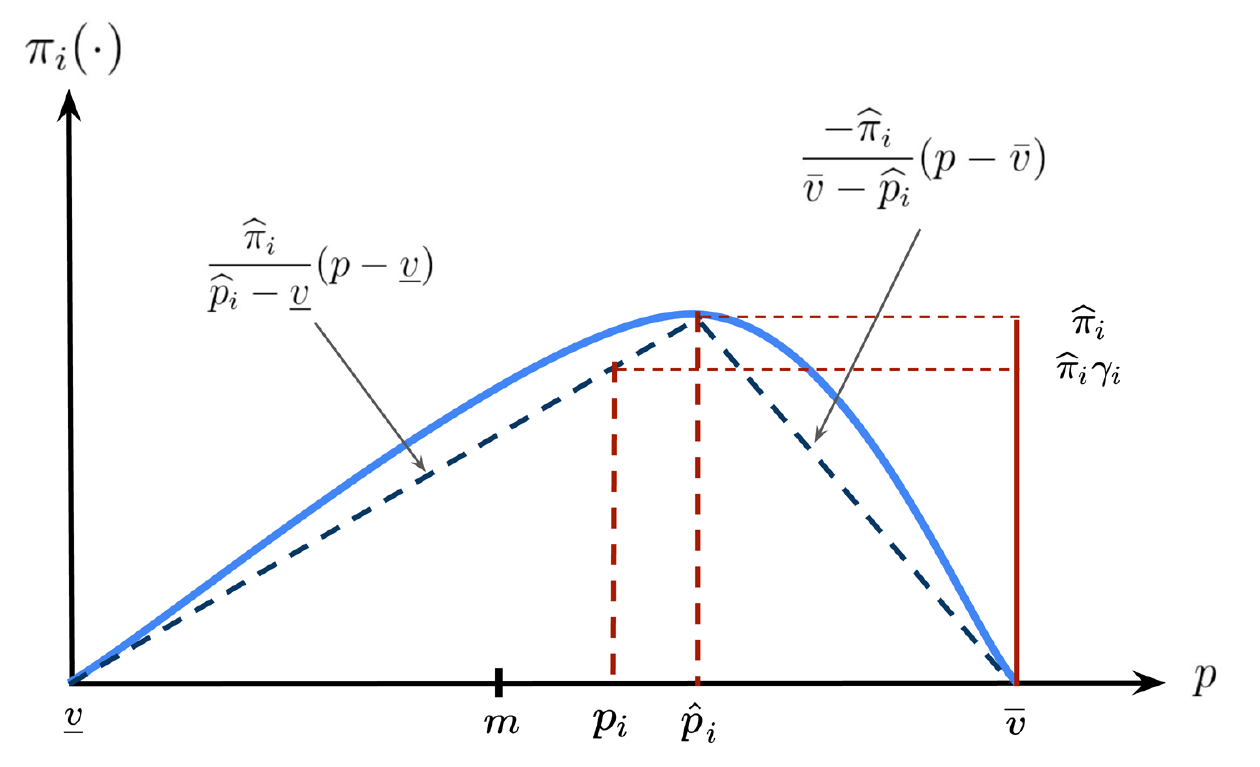}
\caption{Concave revenue function $\pi_i(\cdot)$ and its linear approximation $L_i(\cdot)$ (arrows show equations for $L_i(\cdot)$). Figure represents the case $\widehat{p}_{i} - m \geq \alpha D_i/2$ for which \OurFPP\ assigns $p_i = m + \alpha D_i/2$.  The case $m - \widehat{p}_{i} \geq \alpha D_i/2$ is similar.}
\label{fig:my_label}
\end{figure}

\begin{restatable}{lemma}{CoF} The pricing strategy given in \cref{eq:fairPricing} guarantees at-least $\gamma_i$ fraction of optimal revenue from segment $i$, i.e.,   $\pi_{i} \geq \gamma_i\widehat{\pi}_{i}$.
\label{lem:approx}
\end{restatable}
\begin{proof}
A lower bound to the concave revenue functions $\pi_i(\cdot)$ for any segment $i$ is the piecewise linear approximation $L_i$, given by (see Figure  \ref{fig:my_label}):
\begin{equation}
    L_i(p) = \begin{cases} \frac{\widehat{\pi}_i}{\widehat{p}_i-\underline{v}} (p-\underline{v}), & p \leq \widehat{p}_i \\
    \frac{-\widehat{\pi}_i}{\Bar{v}-\widehat{p}_i}(p-\Bar{v}), & p > \widehat{p}_i
    \end{cases}
\end{equation}
So, for each consumer segment $i$ we have, 
\begin{equation*}\label{eq:linapprox}
L_i(p) \leq \pi_i(p), \ \forall p \in [\underline{v},\Bar{v}]
\end{equation*}

Expected revenues generated per consumer in segment $i$ by pricing rule in Eq. \eqref{eq:fairPricing} for $\widehat{p}_{i} - m \geq \alpha D_i/2$, $m - \widehat{p}_{i} \geq \alpha D_i/2$, and remaining cases are given below in the respective order
\begin{align*}
\pi_i(p_i)& \geq L_i(p_i) 
= \frac{\widehat{\pi}_i}{\widehat{p}_i - \underline{v}}(m+\alpha D_i/2-\underline{v}) = \widehat{\pi}_i \gamma_i \\
\pi_i(p_i)& \geq L_i(p_i) 
= \frac{-\widehat{\pi}_i}{\Bar{v} - \widehat{p}_i}(m - \alpha D_i/2 - \Bar{v} ) = \widehat{\pi}_i \gamma_i \\
\pi_i(p_i)& = L_i(\widehat{p}_i) = \widehat{\pi}_i
\end{align*}
This proves the lemma.
\end{proof}

Let $\pi_{i}^{\star}$ denote the expected revenue generated from the $i^{\text{th}}$ segment by $\mathbf{p}^{\star}$. So, \cof\ for optimal FPP is given by: 
\begin{subequations}
\begin{align*}
    \cof  &= \frac{\sum \limits_{i \in [K]} \beta_{i} \widehat{\pi}_{i}}{\sum \limits_{i \in [K]} \beta_{i} \pi_{i}^{\star}} \leq \frac{\sum \limits_{i \in [K]} \beta_{i} \widehat{\pi}_{i}}{\sum \limits_{i \in [K]} \beta_{i} \pi_{i}} \tag{Optimality of $\pi_i^{\star}$}\\
    & \leq \frac{\sum \limits_{i \in [K]} \beta_{i} \widehat{\pi}_{i} }{\sum \limits_{i \in [K]} \beta_{i} \gamma_i \widehat{\pi}_{i}} \tag{\cref{lem:approx}}
    \end{align*}
\end{subequations}

In order to prove the said \cof\ bound, it suffices to show that there exists an $m$ (and hence a corresponding pricing strategy using \cref{eq:fairPricing}) for which the said \cof\ bound is satisfied. It can be seen that for $m = (\underline{v} + \Bar{v})/2$, and replacing denominators in \cref{eq:gamma} by $\Bar{v}-\underline{v}$, we have that
\begin{subequations}
\begin{align*}
   \cof & \leq  \frac{\sum_{i \in [K]} \beta_{i} \widehat{\pi}_{i}}{\sum_{i \in [K]} \beta_{i} \widehat{\pi}_{i}\left(\frac{1}{2} + \min\{ \frac{\alpha D_i}{2 (\Bar{v} - \underline{v})} , 1 \}\right)} \\
    & \leq  \frac{\sum_{i \in [K]} \beta_{i} \widehat{\pi}_{i}}{\left(\sum_{i \in [K]} \beta_{i} \widehat{\pi}_{i}\right) \left(\frac{1}{2} + \min\{ \frac{\alpha \min_j D_j}{2 (\Bar{v} - \underline{v})} , 1 \}\right)} \\ 
    & = \frac{2}{ 1 + \min \left\{ \alpha \frac{\min_j D_j}{\Bar{v} - \underline{v}}  , 1 \right\}}
\end{align*}
\end{subequations}\end{proof}

It is worth noting here that the cost of fairness does not depend on the number of the segments and the distribution of the population among these segments. So, if the segments are well separated in terms of the distance between features of consumers across segments the number of segments as well as the distribution of consumer population in these segments do not affect revenue guarantee. Also, if the admissible prices are supported over a large interval, the fairness guarantee becomes weaker. This insight discourages  pricing schemes with wildly varying prices across segments.
Finally, if $\alpha = 0$, i.e., without any fairness constraints, we recover the bound of 2 proved in \cite{bergemann2021thirddegree}. 

We emphasize that the bound is strictly less than $2$ because, under fairness constraints, $\alpha \neq 0$ and typically the consumer types are well separated in the feature space according to the metric $d$ else, the consumer types are indistinguishable for the seller hence, $d_{ij} \neq 0$ for all $i, j \in [K]$. This is an improvement of the \cof\ bound given in \cite{bergemann2021thirddegree}.

\paragraph{Tightness of \cof\ bound:}
We claim that the \cof\ bound presented above is tight. In the following example, equality holds and proves the tightness of the bound. 
\begin{example}[Tightness of the \cof\ bound]
Consider $K=2$ where $\beta_1=\beta_2=\frac{1}{2}$. Consider $\mathcal{F}_i$ be such that $\pi_{i}(\cdot) = L_i(\cdot)$ with $\widehat{p}_1 = \underline{v}+\epsilon, \widehat{p}_2 = \Bar{v}-\epsilon,$ where $\epsilon \rightarrow 0$, and $\widehat{\pi}_1 = \widehat{\pi}_2$.
It can be seen that if $\alpha$ is such that $\alpha d_{12} < \Bar{v}-\underline{v}$, any FP satisfying  $p_2-p_1 = \alpha d_{12}$ and $p_1,p_2 \in [\widehat{p}_1, \widehat{p}_2]$ is an optimal FFP (fair FP), and the corresponding 
$\cof = \frac{2}{1+\frac{\alpha d_{12}}{\Bar{v}-\underline{v}}}$. If $\alpha d_{12} \geq \Bar{v}-\underline{v}$, the optimal FP is $\alpha$-fair and so, $\cof = 1$.
Hence, for this example, $\cof = \frac{2}{1 + \min \left\{ \alpha \frac{d_{12}}{\Bar{v} - \underline{v}} , 1 \right\}}$.
This shows the tightness of the \cof\ bound derived in \cref{thm:cofbound}.
\end{example}

We now present an algorithm, OPT-\OurFPP, to find the optimal pivot $m^{\star}$ in the above \OurFPP\ strategy when only $\widehat{p}$ and $\widehat{\pi_i}$s are known. 

\subsection{Proposed Algorithm}
\label{sec:algorithm}
As \OurFPP\ satisfies \ourname (\cref{prop:linp}), and also achieves \cof\ bounds in ~\cref{thm:cofbound}, we look for a pricing strategy optimal within class of \OurFPP. It reduces to finding an optimal pivot that maximizes revenue. In this section, we  propose a binary-search-based algorithm for the same. 
For pricing $\mathbf{p}$, the expected revenue generated per consumer is given by
$
\Pi(\mathbf{p}) = \sum_{i=1}^{K} \beta_i \pi_i(p_i)
$. Let $\tau_i := \frac{\alpha}{2} D_i$. Observe from \cref{lem:approx} that $\Pi(\mathbf{p})$ is lower bounded as:

\begin{subequations}\label{eq:fairrev}
\begin{align*}
    \Pi &(\mathbf{p}) \geq \Pi_m(\mathbf{L}) = \sum_{i=1}^{K} \beta_i \gamma_i \widehat{\pi}_i  = \sum \limits_{\substack{ i:|\widehat{p}_i - m |< \tau_i}}\beta_i \widehat{\pi}_i \ \ + \\&   \sum_{\substack{ i:\widehat{p}_i - m \geq \tau_i}}\beta_i \widehat{\pi}_i \frac{m+ \tau_i-\underline{v}    }{\widehat{p}_i - \underline{v}}  + \sum_{\substack{ i:m-\widehat{p}_i \geq \tau_i}} \beta_i \widehat{\pi}_i \frac{ \Bar{v}-m +  \tau_i  }{\Bar{v} - \widehat{p}_i}  \tag{\ref{eq:fairrev}}
\end{align*}
\end{subequations}

\subsubsection*{Determining Optimal Pivot $m$}
As we can see, the revenue generated by \OurFPP\ is lower bounded by a piecewise linear function in $m$. 
With the aim of achieving a better lower bound, we now address the problem of determining an optimal pivot  $m^{\star} \in \underset{m \in [\underline{v}, \Bar{v}]}{\operatorname{argmax}} \ \Pi_m(\mathbf{L})$.

\subsubsection*{Pricing Algorithm}

In what follows, we call the candidate points $m$ for optimal pivot, i.e.,  for maximizing $\Pi_m(\mathbf{L})$, as \emph{critical points}. We denote the set of these critical points as $\mathcal{M}$.

\begin{lemma}\label{lem:unimodal}
$\Pi_m(\mathbf{L})$ as a function of $m$ is concave and piecewise linear with the set of critical points $\mathcal{M} = \left( \{\widehat{p}_i - \frac{\alpha}{2} D_i , \widehat{p}_i + \frac{\alpha}{2} D_i \}_{i \in [K]} \cap [\underline{v}, \Bar{v}] \right) \cup \{\underline{v}, \Bar{v}\} $. 
\end{lemma}
\begin{proof}
It is easy to see that for a segment $i$, $\gamma_i$ as a function of $m$ is continuous and piecewise linear with breakpoints (i.e., points at which piecewise linear function changes slope): $\widehat{p}_i - \frac{\alpha}{2} D_i$ and $\widehat{p}_i + \frac{\alpha}{2} D_i$ provided they are in the range $[\underline{v}, \Bar{v}]$. The set of breakpoints is hence $\{ \widehat{p}_i - \frac{\alpha}{2} D_i, \widehat{p}_i + \frac{\alpha}{2} D_i\} \cap [\underline{v}, \Bar{v}]$.
Also, the slope monotonically decreases at the breakpoints, i.e., $\gamma_i$ is a concave function of $m$.

From \cref{eq:fairrev}, we can see that $\Pi_m(\mathbf{L})$ is a weighted sum over all segments, of $\gamma_i$'s with constant weights $\beta_i \widehat{\pi}_i$. 
So, $\Pi_m(\mathbf{L})$ as a function of $m$ is concave and piecewise linear with breakpoints belonging to the following set: $\{\widehat{p}_i - \frac{\alpha}{2} D_i , \widehat{p}_i + \frac{\alpha}{2} D_i \}_{i \in [K]} \cap [\underline{v}, \Bar{v}]$.
Hence, a point $m$ that maximizes $\Pi_m(\mathbf{L})$ belongs to either the aforementioned set of breakpoints, or the set of its boundary points $\{\underline{v}, \Bar{v}\}$. 
Thus, the set of critical points $\mathcal{M} = \left( \{\widehat{p}_i - \frac{\alpha}{2} D_i , \widehat{p}_i + \frac{\alpha}{2} D_i \}_{i \in [K]} \cap [\underline{v}, \Bar{v}] \right) \cup \{\underline{v}, \Bar{v}\} $.
\end{proof}

Our algorithm OPT-\OurFPP\ (Optimal Linearized Pivot-based Fair Feature-based Pricing)
which determines an optimal pivot $m^{\star}$ and provides an $\alpha$-fair pricing strategy ($\widetilde{\mathbf{p}}$) is presented in \cref{alg:ouralgo}.

\begin{algorithm}[!ht]
\caption{OPT-\OurFPP\ } 
\label{alg:ouralgo}
\DontPrintSemicolon
  
 \KwInput{$\alpha, \mathbf{\widehat{p}}, (\widehat{\pi}_1,\ldots,\widehat{\pi}_K), (\beta_1,\ldots,\beta_K), (D_1, \ldots, D_K)$}
  \KwOutput{$m^{\star}, \widetilde{\mathbf{p}} $ }

\tcc{Creating and sorting the set of critical points}
$\mathcal{M} \gets \{\underline{v}, \Bar{v}\}$  \;
\For{$i \in [K]$}{
$\tau_i \gets \frac{\alpha}{2} D_i$ \;
\If{$\widehat{p}_i - \tau_i > \underline{v}$}{
$\mathcal{M} \gets  \mathcal{M} \cup \{\widehat{p}_i - \tau_i\}$ \;
}
\If{$\widehat{p}_i + \tau_i < \Bar{v}$}{
$\mathcal{M} \gets  \mathcal{M} \cup \{\widehat{p}_i + \tau_i\}$ \;
}
}
sort($\mathcal{M}$) \;

\tcc{Binary search for optimal pivot}
$\ell \gets 0$, $r \gets |\mathcal{M}|-1$ \;
\While{$\ell \leq r$}
{
$z \gets \lfloor\! \frac{\ell+r}{2} \!\rfloor$
\tcp*{$\!\mathcal{M}[z]\!$ is the current pivot}
\tcc{Computing the expression in \cref{eq:fairrev} at current and adjacent critical points}
$\Pi_{\mathcal{M}[z-1]} \gets 0$, $\Pi_{\mathcal{M}[z]} \gets 0$, $\Pi_{\mathcal{M}[z+1]} \gets 0$ \;
    \For{$y \gets \{z-1, z, z+1\}$}
    {
    \For{$i \gets 1$ to $K$}
    {
    \uIf{$\widehat{p}_{i} \geq \mathcal{M}[y] + \tau_i$}{
    $\gamma_i \gets \frac{\mathcal{M}[y]  -\underline{v} + \tau_i }{\widehat{p}_i - \underline{v}}$
    }
    \uElseIf{$\widehat{p}_{i} \leq \mathcal{M}[y] -  \tau_i$}{
    $\gamma_i \gets \frac{ \Bar{v}- \mathcal{M}[y] + \tau_i  }{\Bar{v} - \widehat{p}_i}$
    }
    \Else{
    $\gamma_i \gets 1$
    }
    $\Pi_{\mathcal{M}[y]} \gets \Pi_{\mathcal{M}[y]} + \beta_i \gamma_i \widehat{\pi}_i$ \;
    }
}
\uIf{$\Pi_{\mathcal{M}[z-1]} \leq \Pi_{\mathcal{M}[z]} \leq \Pi_{\mathcal{M}[z+1]}$}
{
$\ell \gets z+1$ \;
}
\uElseIf{$\Pi_{\mathcal{M}[z-1]} \geq \Pi_{\mathcal{M}[z]} \geq \Pi_{\mathcal{M}[z+1]}$}
{
$r \gets z-1$
}
\Else{
$m^{\star} \gets \mathcal{M}[z]$ \;
\textbf{break} \;
}
}
\tcc{Pricing for the different segments}
\For{$i \in [K]$}{
\uIf{$\widehat{p}_{i}  \geq m^{\star} + \tau_i$}{
$\widetilde{p}_i \gets m^{\star} + \tau_i$
}
\uElseIf{$\widehat{p}_{i}  \leq m^{\star} - \tau_i$}{
$\widetilde{p}_i \gets m^{\star} - \tau_i$
}
\Else{
$\widetilde{p}_i \gets \widehat{p}_i$
}
}
\end{algorithm}

\begin{theorem}
The \emph{OPT}-\OurFPP\ algorithm   (a) returns optimal pivot point $m^{\star}$ and runs in $\mathcal{O}(K \log(K))$ time, and  (b) achieves the  \cof\  bound given in \cref{thm:cofbound}. 
\label{thm:opt_cofbound}
\end{theorem}
\begin{proof}
(a) The first module is the creation and sorting of the set of critical points $\mathcal{M}$, which takes $\mathcal{O}(K \log(K))$ time.
Owing to \cref{lem:unimodal}, we can find an optimal pivot $m^{\star}$ using binary search over $\mathcal{M}$. Here, the number of critical points are at most $2K+2$, i.e., $|\mathcal{M}| \leq 2K+2$.
So, in the second module that finds an optimal pivot, the binary search in the outer (\emph{while}) loop runs for $\mathcal{O}(\log(|\mathcal{M}|))$ iterations, and the inner (\emph{for}) loops run for $\mathcal{O}(K)$ iterations overall.
Thus, the running time of the second module is $\mathcal{O}(K \log(K))$.
The third module that computes pricing for the different segments runs in $\mathcal{O}(K)$ time.
So, the total running time of Algorithm \ref{alg:ouralgo} is $\mathcal{O}(K \log(K))$. 

(b) From \cref{thm:cofbound}, for $m=(\underline{v}+\Bar{v})/2$, the \cof\  bound holds. Also, $\Pi_{m^{\star}}(\mathbf{L}) \geq \Pi_m(\mathbf{L})$ for all $m \neq m^{\star}$. We have:
$$
\cof = \frac{\Pi(\widehat{\mathbf{p}})}{\Pi(\widetilde{\mathbf{p}})} \leq \frac{\Pi(\widehat{\mathbf{p}})}{\Pi_{m^{\star}}(\mathbf{L})} \leq \frac{\Pi(\widehat{\mathbf{p}})}{\Pi_m(\mathbf{L})}   
$$
This completes the proof of the theorem.
\end{proof}

\section{Discussion}
This paper built a foundation for the design of fair feature-based pricing by proposing a new fairness notion called \ourname$\!\!$. Our impossibility result on the discrete valuation setting restricted us from attaining a finite cost of fairness (\cof) in general settings. Interestingly, in the continuous valuation setting with concave revenue functions, we showed that a family of pricing schemes, \OurFPP, provided a \cof\ strictly less than $2$. Finally, we proposed an algorithm, OPT-\OurFPP, which gave us an optimal pricing strategy within this family. It is worth noting that the algorithm does not require a complete distribution function; peaks of revenue distributions are sufficient statistics for computing optimal fair feature-based pricing. 

We leave the problem of finding an optimal segmentation (optimal value of $K$ and corresponding $K$-partition of the market) as interesting future work. We assumed a monopoly market. It will be interesting to study optimal fair pricing in the face of competition and other constraints such as finite supply, non-linear   production cost, 
and variable demand.

\bibliographystyle{ACM-Reference-Format}
\bibliography{arxiv}
\end{document}